\documentclass[english,12pt]{article}
\usepackage[latin9]{inputenc}
\usepackage{geometry}
\usepackage{float}
\usepackage{amsthm}
\usepackage{amsmath}
\usepackage{amssymb}
\usepackage{graphicx}
\usepackage{setspace}
\doublespace
\usepackage{esint}  
\usepackage[authoryear]{natbib}
\usepackage{hyperref}
\usepackage{comment}
\usepackage[normalem]{ulem}
\usepackage{array}
\usepackage{multirow}
\usepackage{float}
\usepackage{cancel}
\usepackage{enumerate}
\usepackage{pgf,tikz}
\usepackage{mathrsfs}
\usetikzlibrary{arrows}
\usepackage{babel}
\usepackage{bbm}

\theoremstyle{plain}
\newtheorem{theorem}{\protect\theoremname}
\theoremstyle{plain}
\newtheorem*{theorem*}{\protect\theoremname}
\theoremstyle{definition}
\newtheorem{definition}{\protect\definitionname}
\theoremstyle{definition}
\newtheorem*{definition*}{\protect\definitionname}
\theoremstyle{assumption}

\theoremstyle{plain}
\newtheorem{corollary}{\protect\corollaryname}
\theoremstyle{plain}
\newtheorem{lemma}{\protect\lemmaname}
\theoremstyle{plain}
\newtheorem{proposition}{\protect\propositionname}
\theoremstyle{plain}
\newtheorem{conjecture}{\protect\conjecturename}
\theoremstyle{definition}%%%12/28/2019

\theoremstyle{plain}

\theoremstyle{plain}

\theoremstyle{plain}

\theoremstyle{plain}

\providecommand{\claimname}{Claim}
\providecommand{\conjecturename}{Conjecture}
\providecommand{\corollaryname}{Corollary}
\providecommand{\lemmaname}{Lemma}
\providecommand{\definitionname}{Definition}
\providecommand{\assumptionname}{Assumption}

\providecommand{\factname}{Fact}
\providecommand{\propositionname}{Proposition}
\providecommand{\remarkname}{Remark}
\providecommand{\theoremname}{Theorem}
\providecommand{\resultname}{Result}
\providecommand{\observationname}{Observation}

\usepackage{subcaption}
\makeatletter
\renewcommand{\fnum@figure}{Figure \thefigure}
\makeatother

\usepackage[font=footnotesize]{caption} 

\usepackage{physics}
\usepackage{enumitem}
\usepackage{pgfplots}
\usepackage{mathtools}

\DeclareMathOperator*{\argmax}{arg\,max}
\DeclareMathOperator*{\argmin}{arg\,min}
\DeclareMathOperator{\supp}{Supp}
\DeclareMathOperator{\proj}{Proj}

\usepackage{istgame}
\usetikzlibrary{shapes.geometric,arrows}

\usepackage{xcolor}
\definecolor{myblue}{RGB}{0,8,126}
\definecolor{mygreen}{RGB}{1,117,25}
\definecolor{mygrey}{RGB}{75,77,101}
\definecolor{mypurple}{RGB}{116,1,116}
\definecolor{mylightgrey}{RGB}{214,214,214}
\definecolor{mywinered}{RGB}{183,0,56}
\definecolor{mybrown}{RGB}{183,116,56}
\definecolor{myorange}{RGB}{255,116,1}

\newcommand{\pdfcolor}{blue!85!black}
\hypersetup{
    colorlinks=true,
    linkcolor=\pdfcolor,
    citecolor=\pdfcolor,
    urlcolor=\pdfcolor}

\begin{document}

%This note starting from April 25, 2023
\title{Private Experimentation, Data Truncation, and Verifiable Disclosure}

\author{Yichuan Lou\thanks{Department of Economics, The University of Tokyo. E-mail: \url{yichuanlou@e.u-tokyo.ac.jp}.}\\ \textbf{\small preliminary and incomplete}
}
\date{\today}

\maketitle
\thispagestyle{empty}
\begin{abstract}
\noindent
A sender seeks to persuade a receiver by presenting evidence obtained through a sequence of private experiments. The sender has complete flexibility in his choice of experiments, contingent on the private experimentation history. The sender can disclose each experiment outcome credibly, but cannot prove whether he has disclosed everything. By requiring `continuous disclosure', I first show that the private sequential experimentation problem can be reformulated into a static one, in which the sender chooses a single signal to learn about the state. Using this observation, I derive necessary conditions for a signal to be chosen in equilibrium, and then identify the set of beliefs induced by such signals. Finally, I characterize sender-optimal signals from the concavification of his value function constrained to this set.

\end{abstract}

\newpage
\section{Introduction}

This paper deals with an situation in which one agent (the sender) seeks evidence to persuade another agent (the receiver) to take certain actions. The evidence is generated by a series of private experiments where the sender can flexibly design each experiment contingent on the history. The sender can disclose each piece of evidence credibly, but cannot prove whether he has disclosed everything. The sender has a choice of concealing information through `right truncation' or `right censoring': disclosing all previous evidence is necessary to reveal a specific piece of evidence, but it is possible to keep any subsequent evidence concealed. The restriction to right-truncated data is also called `continuous disclosure' as the sender cannot cherry-pick outcomes in a discontinuous manner.

There are various examples in real life where data must be disclosed continuously from time $0$ to a certain cutoff, without the ability to selectively disclose certain data points. One example is in financial reporting, where a company chooses to release information about its financial performance. In this case, the company may disclose all relevant financial information up until certain point but not necessarily any information that came afterward. Similarly, in scientific research, a researcher collects and analyzes data, and may only report on data that was collected up to some endogenous cutoff date, disregarding any data collected after that point. In both cases, the requirement for continuous disclosure is intended to mitigate the risk of data manipulation.

This paper focuses on the equilibrium outcomes of such evidence collection and disclosure games. I study the type of information that is generated and disclosed in equilibrium. I also characterize optimal evidence collection from the sender's perspective.

With this goal in mind, I consider a communication game between a sender (he) and a receiver (she). Initially, the sender is uninformed and shares the same prior belief as the receiver. The sender covertly acquires information about the state of the world. In particular, he may run a series of experiments and tailors the characteristics of each experiment contingent on past experimentation. After choosing to stop experimenting, the sender can reveal experiment outcomes in the manner of right truncation as mentioned above. The receiver observes the sender's message, and then she takes an action.

The main results of the paper are as follows:  

(A) By requiring `continuous disclosure', I observe that the private sequential experimentation problem can be reformulated into a static one, in which the sender chooses a single signal to learn about the state, without a second chance. %and stops experimenting regardless of the realization. 

(B) The condition that the sender has no incentive to acquire further information reduces the set of beliefs that can be induced in equilibrium. This insight allows me to take a belief-based approach, commonly used in the communication literature. More specifically, the sender's payoff under the optimal signal is precisely the concavification of his utility function across those beliefs, evaluated at the prior.

\paragraph{Related Literature.} The strategic/voluntary disclosure of verifiable information goes back to \cite{GrossmanHart:80}, \cite{Grossman:81}, and \cite{Milgrom:81}, who establish the `unraveling' result that in any equilibrium Sender fully reveals his information---if some types pool, at least one of them is `better' than the `average' and will prefer to reveal himself. The result relies on certain assumptions including common knowledge that Sender is exogenously and privately informed, disclosing is costless, and information is verifiable. In contrast to this literature, I show that if private information is endogenous and gathered covertly, then essentially all information will be revealed in equilibrium; if there is an equilibrium in which information is withheld, the outcome must be the same as in another equilibrium with full revelation.\footnote{The covert nature of information acquisition is not essential to full revelation. \cite{GentzkowKamenica:17b}, for example, establish a similar result in the overt case.}

The paper contributes to a few strands of the persuasion literature in which the production of evidence is endogenous. \cite{BrocasCarrillo:07}, \cite{FelgenhauerSchulte:14}, \cite{FelgenhauerLoerke:17}, \cite{HenryOttaviani:19}, and \cite{Herresthal:22} consider environments in which information acquisition is sequential: Sender decides at each instant whether to continue or to stop experimenting. In a framework with symmetric information (i.e., both experiments and outcomes are publicly observable), \cite{BrocasCarrillo:07} and \cite{HenryOttaviani:19} show that when Sender controls the flow of public information, the structure of the solution is closely related to the one characterized in \cite{KamenicaGentzkow:11}. Assuming experimentation is covert (both experiments and outcomes are private information), \cite{FelgenhauerSchulte:14}, \cite{FelgenhauerLoerke:17} and \cite{Herresthal:22} study the incentives to collect information sequentially and to disclose privately observed results. By mandating `continuous disclosure' and limiting Sender's ability to certify what experiments were conducted, a contribution of my paper is the observation that the dynamic game can be equivalently transformed into a static one amenable to simple analysis.

Another strand of the literature considers environments in which information acquisition is one-shot; Sender is constrained to experiment only once. \cite{GentzkowKamenica:17b} study overt acquisition of evidence in a disclosure model where each type can perfectly self-certify and show that disclosure requirements never change the set of equilibrium outcomes. \cite{DeMarzoKremerSkrzypacz:19} and \cite{Shishkin:23} endogenize Sender's evidence in the voluntary disclosure model of \cite{Dye:85} and \cite{JungKwon88} with positive probability of a null outcome. \cite{Henry:09} studies private experimentation where Sender ex ante chooses how much costly research to perform in both overt and covert cases. He shows that Sender maybe induced to conduct extra research in the latter case, to counteract Receiver's inference that disclosure is selective. \cite{Escude:23} addresses a comparative statics question of how the incentives to acquire and disclose information depend on the verifiability of acquired information.

\section{A Motivating Example}

Consider an assistant professor that negotiates a salary raise with the dean. The dean is contemplating whether to give a small pay raise or a big pay raise, or to maintain the status quo. The dean would like to set the salary depending on the professor's research, which is either \textit{good} or \textit{bad}. The dean will give a big pay raise whenever she assigns a probability of at least $4/5$ to the research being \textit{good}, and will maintain the status quo whenever she assigns a probability of at least $3/5$ to the research being \textit{bad}. For intermediate beliefs, the dean offers a small pay raise. On the other hand, the professor gets utility $0$ if he does not get a pay raise, utility $1$ if he gets a small pay raise, utility $6/5$ if he gets a big pay raise, regardless of the state of the world---his research. The professor and the dean share a prior belief of $Pr(good) = 1/5$.

To persuade the dean to give a raise, the professor provides evidence of his research that stems from experimentation. One can think of the choice of experiments as consisting of decisions on whether to give a departmental seminar about his current research, which journal to submit a recent paper to, whether to go to the senior job market and have some outside offers, and so on. I formalize an experiment as distributions $\lambda(\cdot|good)$ and $\lambda(\cdot|bad)$ on some set of experimental outcomes.

Consider first public experimentation: the experimentation history is common knowledge. The professor privately observes the experimental outcomes and then discloses them strategically via verifiable messages. \cite{KamenicaGentzkow:11} show that this formulation of public acquisition of private information, yields equilibrium outcomes that are identical to those that arise in their full commitment model which additionally requires the professor to disclose his private information truthfully. Using their results, one can deduce that the professor optimally chooses a \textit{single} binary experiment, denoted by $\lambda_1$:
\begin{alignat*}{2}
    \lambda_1(g_1|good) &= 1\quad \lambda_1(g_1|bad)&&=3/8\\
    \lambda_1(b_1|good) &= 0\quad \lambda_1(b_1|bad)&&=5/8.
\end{alignat*}
This leads the dean to give a small pay raise with probability $50$ and to maintain the status quo with probability $50$.

Next, consider private experimentation: the professor privately and sequentially runs as many experiments (depending on what he has learned from previous experimental outcomes) as desired and selectively reveals the results. The professor's experimentation history after the first $t$ experiments is denoted by $h_t = \{(\lambda_i,y_i)\}_{i=1}^t$, where $y_i$ denotes the $i^{th}$ experimental outcome of $\lambda_i$. I assume the professor cannot manipulate or fabricate experimental outcomes. In this sense, $(\lambda_i,y_i)$ is `hard' information. However, the professor can conceal experimental outcomes in a `chronological' way: if he wants to conceal the $i^{th}$ experimental outcome, $(\lambda_i,y_i)$, he has to conceal all outcomes after the $i^{th}$ experiment. In other words, the professor's message takes a \textit{right-truncated} form---in order to reveal some experimental outcome, he must reveal all precedent outcomes as well. To illustrate, suppose the professor first submitted his paper to a top-tier journal and got rejected, and then submitted it to a second-tier journal and got accepted. Then, the dean is entitled to know that the professor received a rejection from the top-tier journal first if the professor wants to reveal that his paper got accepted.

It is by now well-known that Bayesian persuasion (or equivalently, public experimentation here) establishes an upper bound on Sender's gain from any possible communication protocols. This naturally leads to a question of whether the professor can achieve this upper bound as well when experimentation becomes private. The answer turns out to be: No. To see this, suppose the dean expects the professor to conduct $\lambda_1$, to stop after observing either realization, and to truthfully communicate the experimental outcome. If the outcome of $\lambda_1$ is $b_1$, then the professor will act as expected by stopping immediately and reporting $b_1$. This is because the professor cannot benefit by continuing experimenting and hopefully obtaining some positive outcome to reveal---he is required to report $b_1$ which leads to the posterior belief $Pr(good)=0$. By contrast, if the outcome of $\lambda_1$ is $g_1$, the professor could secretly deviate by running a second experiment, denoted by $\lambda_2$:
\begin{alignat*}{2}
    \lambda_2(g_2|good) &= 1\quad \lambda_2(g_2|bad)&&=0\\
    \lambda_2(b_2|good) &= 0\quad \lambda_2(b_2|bad)&&=1.
\end{alignat*}
That is, the professor additionally runs a fully informative experiment, one that leaves no uncertainty about the state. Then, if the outcome of $\lambda_2$ is $g_2$, the professor can `surprise' the dean by revealing $(\lambda_2,g_2)$ appended to $(\lambda_1,g_1)$, which will convince the dean that the research is $good$ with probability $1$ and will deliver the highest payoff of $6/5$. If the outcome of $\lambda_2$ is $b_2$, the professor can report $(\lambda_1,g_1)$ only, which still secures himself a payoff of $1$. As a result, the revealed evidence $(\lambda_1,g_1)$ will not be taken \textit{at face value} by the dean knowing that the professor will run an additional private experiment and reveals $(\lambda_1,g_1)$ if the second experiment fails.

The professor can, however, overcome the dean's skepticism by running the following single experiment, denoted by $\lambda_1'$:
\begin{alignat*}{2}
    \lambda_1'(g_1|good) &= 1\quad \lambda_1'(g_1|bad)&&=1/16\\
    \lambda_1'(b_1|good) &= 0\quad \lambda_1'(b_1|bad)&&=15/16.
\end{alignat*}
In fact, it is easy to see that under the experiment $\lambda_1'$ the professor no longer has an incentive to deviate by running additional private experiments and concealing unfavorable outcomes. This leads the dean to give a big pay raise with probability $40$ percent and to maintain the status quo with probability $60$ percent. As I show below, this is the best the professor can do from sequential private experimentation.

Figure \ref{fig:professor_dean} visualizes the example. Because the state is binary, I identify the dean's posterior belief $\mu$ with the probability it assigns to the professor's research being \textit{good}, $Pr(good)$. Putting this probability on the horizontal axis, the figure plots the highest value the professor can obtain from uninformative communication, Bayesian persuasion (equivalently, public experimentation), and private sequential experimentation. That is, the figure plots the professor's value function (left), along with his concave envelope (center) and concave envelope over the red shaded region (right). This red shaded region, as I will show below, represents all `additional-learning-proof' beliefs given which the professor prefers truthful revelation instead of generating further information followed by selective revelation.

\vspace{1em}
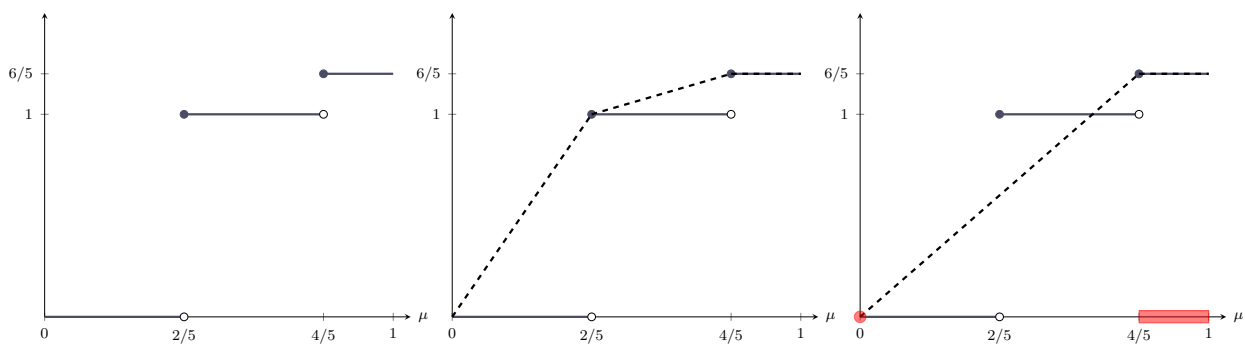
\begin{figure}[H]
\centering
\captionsetup{font=footnotesize}
\begin{minipage}{0.32\linewidth}
\begin{tikzpicture}[scale=.71]
    \begin{axis}
    [
    axis lines = left,
    xlabel={$\mu$},
    ylabel={},
    xmin=0, xmax=1.05,
    ymin=0, ymax=1.5,
    xtick={0,2/5,4/5,1},
    xticklabels={$0$,$2/5$,$4/5$,$1$},
    ytick={1,6/5},
    yticklabels={$1$,$6/5$},
    every axis x label/.style={at={(ticklabel* cs:1)},anchor=west,},
    scaled y ticks = false,every axis y label/.style={at={(ticklabel* cs:1)},anchor=south,},
    label style={font=\scriptsize},
    tick label style={font=\scriptsize},
    clip=false
    ]
    \draw[very thick,mygrey](axis cs:0,0)--(axis cs:2/5,0);
    \draw[very thick,mygrey](axis cs:2/5,1)--(axis cs:4/5,1);
    \draw[very thick,mygrey](axis cs:4/5,6/5)--(axis cs:1,6/5);
    
    \draw[fill=white] (axis cs:2/5,0) circle (.4ex);
    \draw[fill=white] (axis cs:4/5,1) circle (.4ex);
    \draw[fill=mygrey, mygrey] (axis cs:2/5,1) circle (.4ex);
    \draw[fill=mygrey, mygrey] (axis cs:4/5,6/5) circle (.4ex);
    \end{axis}
\end{tikzpicture}
\end{minipage}
\begin{minipage}{0.32\linewidth}
\begin{tikzpicture}[scale=.71]
    \begin{axis}
    [
    axis lines = left,
    xlabel={$\mu$},
    ylabel={},
    xmin=0, xmax=1.05,
    ymin=0, ymax=1.5,
    xtick={0,2/5,4/5,1},
    xticklabels={$0$,$2/5$,$4/5$,$1$},
    ytick={1,6/5},
    yticklabels={$1$,$6/5$},
    every axis x label/.style={at={(ticklabel* cs:1)},anchor=west,},
    scaled y ticks = false,every axis y label/.style={at={(ticklabel* cs:1)},anchor=south,},
    label style={font=\scriptsize},
    tick label style={font=\scriptsize},
    clip=false
    ]
    \draw[very thick,mygrey](axis cs:0,0)--(axis cs:2/5,0);
    \draw[very thick,mygrey](axis cs:2/5,1)--(axis cs:4/5,1);
    \draw[very thick,mygrey](axis cs:4/5,6/5)--(axis cs:1,6/5);
    
    \draw[fill=white] (axis cs:2/5,0) circle (.4ex);
    \draw[fill=white] (axis cs:4/5,1) circle (.4ex);
    \draw[fill=mygrey, mygrey] (axis cs:2/5,1) circle (.4ex);
    \draw[fill=mygrey, mygrey] (axis cs:4/5,6/5) circle (.4ex);
    
    \draw[very thick,dashed](axis cs:0,0)--(axis cs:2/5,1);
    \draw[very thick,dashed](axis cs:2/5,1)--(axis cs:4/5,6/5);
    \draw[very thick,dashed](axis cs:4/5,6/5)--(axis cs:1,6/5);
    \end{axis}
\end{tikzpicture}
\end{minipage}
\begin{minipage}{0.32\linewidth}
\begin{tikzpicture}[scale=.71]
    \begin{axis}
    [
    axis lines = left,
    xlabel={$\mu$},
    ylabel={},
    xmin=0, xmax=1.05,
    ymin=0, ymax=1.5,
    xtick={0,2/5,4/5,1},
    xticklabels={$0$,$2/5$,$4/5$,$1$},
    ytick={1,6/5},
    yticklabels={$1$,$6/5$},
    every axis x label/.style={at={(ticklabel* cs:1)},anchor=west,},
    scaled y ticks = false,every axis y label/.style={at={(ticklabel* cs:1)},anchor=south,},
    label style={font=\scriptsize},
    tick label style={font=\scriptsize},
    clip=false
    ]
    \draw[very thick,mygrey](axis cs:0,0)--(axis cs:2/5,0);
    \draw[very thick,mygrey](axis cs:2/5,1)--(axis cs:4/5,1);
    \draw[very thick,mygrey](axis cs:4/5,6/5)--(axis cs:1,6/5);
    
    \draw[fill=white] (axis cs:2/5,0) circle (.4ex);
    \draw[fill=white] (axis cs:4/5,1) circle (.4ex);
    \draw[fill=mygrey, mygrey] (axis cs:2/5,1) circle (.4ex);
    \draw[fill=mygrey, mygrey] (axis cs:4/5,6/5) circle (.4ex);
    
    \draw[very thick,dashed](axis cs:0,0)--(axis cs:4/5,6/5);
    \draw[very thick,dashed](axis cs:4/5,6/5)--(axis cs:1,6/5);
    
    \draw[red,fill=red,opacity=.5] (axis cs:0,0) circle (.6ex);
    \draw[fill=red,red,opacity=.5] (axis cs:4/5,-.03) rectangle (axis cs:1,.03);
    \end{axis}
\end{tikzpicture}
\end{minipage}
\caption{The professor-dean example.}
\label{fig:professor_dean}
\end{figure}

\section{The Model}

I analyze a game with two players: Sender (he) and Receiver (she). Both players' payoffs depend on Receiver's action $a\in A$ and an unknown state of the world $\omega\in \Omega$. Thus, Sender and Receiver have utility functions $v:A\rightarrow \mathbb{R}$ and $u:A\times \Omega\rightarrow \mathbb{R}$, respectively, and each aims to maximize expected payoffs.\footnote{The assumption of Sender's state-independent preferences is common in the literature on communication with hard evidence (e.g., \cite{GlazerRubinstein:04}, \cite{HartKremerPerry:17}, \cite{Rappoport:22}).}

I impose some technical restrictions on my model. Sender and Receiver share a prior $\mu_0\in \Delta(\Omega)$. Each of $\Omega$ and $A$ is a compact metrizable space. Both players' utility functions are continuous. At any given belief about $\omega$, I assume that Receiver has a unique optimal action, i.e., $a^*(\mu)\equiv \argmax_{a\in A}E_{\mu}[u(a,\omega)]$ is single-valued for all $\mu\in \Delta(\Omega)$. When Receiver holds some belief $\mu$, Sender's utility is
\begin{align*}
    \hat{v}(\mu)\equiv v(a^*(\mu)).
\end{align*}

\subsection{The Original Game}

\paragraph{Experiments.} A (Blackwell) experiment $\lambda$ consists of a sufficiently large experimental outcome space $Y$ and a family of distributions $\{\lambda(\cdot|\omega)\}_{\omega\in \Omega}$ over $Y$. Let $\Lambda$ denote the set of all experiments.

\paragraph{The Original Game.} The original game is a sequential game of private experimentation. Sender has access to all experiments in $\Lambda$. He can conduct as many experiments as he wants. Conditional on the state, experimental outcomes are drawn independently. The experimentation history after the first $t$ experiments is denoted by $h_t = \{(\lambda_i,y_i)\}_{i=1}^t$. At each $h_t$, Sender may either continue experimenting and run an additional experiment, or he may stop experimenting and send a message $n$. The set of feasible messages will be specified in the next paragraph. Receiver observes the message $n$, but she cannot observe the experimentation history at which Sender stops experimenting. Receiver then takes her action.

\paragraph{Truncation Messaging.} Sender cannot manipulate or make up experimental outcomes---each outcome is `hard' information. However, he can conceal experimental outcomes through `right truncation': disclosing all previous evidence is necessary to reveal a specific piece of evidence, but it is possible to keep any subsequent evidence concealed. Formally, let $n = \{(\lambda_i,y_i)\}_i$ denote Sender's message.\footnote{Note that Receiver does not only observe the outcomes contained in a message, but also the precision of the experiments with which these outcomes were generated.} A message $n = \{(\lambda_i,y_i)\}_i$ is \textit{$($right-$)$truncation-feasible} given a history $h_t=\{(\lambda_i,y_i)\}_{i=1}^t$, if $n\subset h_t$ and, moreover, $(\lambda_j,y_j)\in n$ implies $(\lambda_i,y_i)\in n$ for all $1\leq i\leq j$. Notationally, let $N(h_t)$ be the set of truncation-feasible messages at $h_t$. I assume that $N(h_t)$ always contains the empty history, denoted by $\emptyset$.

Since Sender's message takes the form of a right-truncated experimentation history, I sometimes refer to such disclosure requirement as `continuous disclosure'.

\paragraph{Strategies and Equilibrium.} Sender's strategy specifies his behavior at each experimentation history $h_t$ that he may observe. At each $h_t$, Sender may either continue experimenting with a further experiment with a history-dependent precision or he may stop experimenting and send his message. Receiver's strategy specifies an action for each message that she may observe. The equilibrium concept is the notion of weak sequential equilibrium in pure strategies. It is the equilibrium that satisfies sequential rationality and weak belief consistency. More details will be given in the following dual formulation to avoid repetition.

\subsection{The Alternate Game}

In this section, I examine an alternative environment in which Sender \textit{ex ante} chooses how much information to gather in one shot, making the acquisition of information static. The rest of the game remains unchanged, but an important issue remains unspecified: if Sender secretly deviates to an experiment with a different precision and different experiment outcomes, how does he pretend that he hasn't deviated when he sends his message? 

To specify such a disclosure game with one-shot evidence collection, I need to reformulate the definition of the set of all experiments in a way that will eliminate such ambiguity. The following formulation does the job.

\paragraph{Signals.} A signal $\pi$ is a partition of the \textit{expanded} state space $\Omega\times[0,1]$ s.t. $\pi\subset S$, where $S$ is the set of non-empty, Lebesgue measurable subsets of $\Omega\times [0,1]$ (\citealt{GreenStokey:78}; \citealt{GentzkowKamenica:17a,GentzkowKamenica:17b}). An element $s\in S$ is a signal realization. The interpretation of this formalism is that a random variable $x$, drawn uniformly from $[0,1]$, determines the signal realization conditional on the state; the conditional probability of $s$ given $\omega$ is $Pr(s|\omega) = \ell(\{x|(\omega,x)\in s\})$ where $\ell(\cdot)$ denotes the Lebesgue measure. Denote by $\mu_s$ the posterior belief about $\omega$ conditional on $s$.\footnote{For any $s$ with $Pr(s) = \sum_{\omega\in \Omega}Pr(s|\omega)\mu_0(\omega)>0$, we have $\mu_s(\omega) = \frac{Pr(s|\omega)\mu_0(\omega)}{Pr(s)}$.}

Let $\Pi$ be the set of all signals. Figure \ref{fig:illustration} illustrates the definition of a signal. In this example, $\Omega=\{L,R\}$ and $\pi = \{l,r\}$ where $l = (L,[0,0.75])\cup (R,[0,0.25])$ and $r = (L,[0.75,1])\cup (R,[0.25,1])$. Thus the signal $\pi$ is a partition of $\Omega\times [0,1]$, and the state-specific likelihood of signal realizations is $Pr(l|L) = Pr(r|R) = 0.75$.

Say $\pi'$ \textit{refines} an element $s$ of $\pi$ if there exists a subset of $\pi'$, denoted by $\mathcal{P}$, such that $s = \cup_{s'\in \mathcal{P}}s'$. Refer to any such $s'$ as a \textit{lower element} of $s$. I then say that $\pi'$ is a refinement of $\pi$, denoted by $\pi'\trianglerighteq\pi$, if $\pi'$ refines every element of $\pi$.\footnote{The refinement order implies the Blackwell informativeness order (\citealt{GreenStokey:78}; \citealt{BrooksFrankelKamenica:22}).} Figure \ref{fig:illustration} illustrates such a signal $\pi'$.

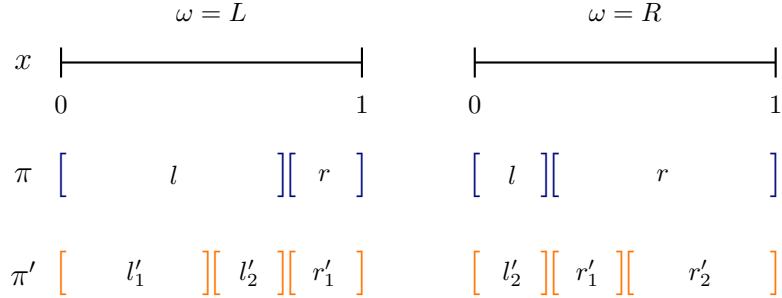
\begin{figure}[H]
\centering
\begin{tikzpicture}[mydrawstyle/.style={draw=black, thick}, x=1mm, y=1mm, z=1mm]
    \draw (-10,70) node {$x$};
    \draw[mydrawstyle](-5,70)--(35,70) node at (15,72)[above=2mm]{\footnotesize$\omega = L$};
    \draw[mydrawstyle](-5,68)--(-5,72) node[below=5mm]{\footnotesize$0$};
    \draw[mydrawstyle](35,68)--(35,72) node[below=5mm]{\footnotesize$1$};
  
    \draw (-10,55) node {$\pi$};
    \draw[very thick, myblue] (-5,55) node[scale=1.4] {[};
    \draw[very thick, myblue] (24.5,55) node[scale=1.4] {]};
    \draw (10,55) node {\footnotesize$l$};
    \draw[very thick, myblue] (25.5,55) node[scale=1.4] {[};
    \draw[very thick, myblue] (35,55) node[scale=1.4] {]};
    \draw[draw=none] (25.5,52.15) rectangle (35,57.85);
    \draw (30,55) node {\footnotesize$r$};
    
    \draw (-10,42) node {$\pi'$};
    \draw[very thick, myorange] (-5,42) node[scale=1.4] {[};
    \draw[very thick, myorange] (14.5,42) node[scale=1.4] {]};
    \draw (5,42) node {\footnotesize$l_1'$};
    \draw[very thick, myorange] (15.5,42) node[scale=1.4] {[};
    \draw[very thick, myorange] (24.5,42) node[scale=1.4] {]};
    \draw[draw=none] (15.5,39.15) rectangle (24.5,44.85);
    \draw (20,42) node {\footnotesize$l_2'$};
    \draw[very thick, myorange] (25.5,42) node[scale=1.4] {[};
    \draw[very thick, myorange] (35,42) node[scale=1.4] {]};
    \draw[draw=none] (25.5,39.15) rectangle (35,44.85);
    \draw (30,42) node {\footnotesize$r_1'$};

    \draw[mydrawstyle](50,70)--(90,70) node at (70,72)[above=2mm]{\footnotesize$\omega = R$};
    \draw[mydrawstyle](50,68)--(50,72) node[below=5mm]{\footnotesize$0$};
    \draw[mydrawstyle](90,68)--(90,72) node[below=5mm]{\footnotesize$1$};
    
    \draw[very thick, myblue] (50,55) node[scale=1.4] {[};
    \draw[very thick, myblue] (59.5,55) node[scale=1.4] {]};
    \draw (55,55) node {\footnotesize$l$};
    \draw[very thick, myblue] (60.5,55) node[scale=1.4] {[};
    \draw[very thick, myblue] (90,55) node[scale=1.4] {]};
    \draw[draw=none] (60.5,52.15) rectangle (90,57.85);
    \draw (75,55) node {\footnotesize$r$};
    
    \draw[very thick, myorange] (50,42) node[scale=1.4] {[};
    \draw[very thick, myorange] (59.5,42) node[scale=1.4] {]};
    \draw (55,42) node {\footnotesize$l_2'$};
    \draw[very thick, myorange] (60.5,42) node[scale=1.4] {[};
    \draw[very thick, myorange] (69.5,42) node[scale=1.4] {]};
    \draw[draw=none] (60.5,39.15) rectangle (69.5,44.85);
    \draw (65,42) node {\footnotesize$r_1'$};
    \draw[very thick, myorange] (70.5,42) node[scale=1.4] {[};
    \draw[very thick, myorange] (90,42) node[scale=1.4] {]};
    \draw[draw=none] (70.5,39.15) rectangle (90,44.85);
    \draw (80,42) node {\footnotesize$r_2'$};    
\end{tikzpicture}
\caption{Signals as partitions of $\Omega\times [0,1]$.}
\label{fig:illustration}
\end{figure}

\paragraph{Timing of the Alternate Game.} The game begins with Sender choosing a signal $\pi$, the choice of which is not observed by Receiver. Sender privately observes a realization $s$ (referred to as Sender's `type') and then sends a message $m\in M^{\pi}(s)$ to Receiver, where $M^{\pi}(s)$ is the set of messages available to the type $s$ given $\pi$ (described below). Receiver observes the message and chooses an action. Note that Sender, albeit granted complete flexibility over what evidence to collect, can only acquire a single signal. This makes his experimentation problem static.%\footnote{For models of sequential information acquisition and verifiable disclosure, see, for example, \cite{FelgenhauerSchulte:14}, \cite{FelgenhauerLoerke:17}, and \cite{Herresthal:22}.}

\paragraph{Evidence Structures.} Without loss of generality, I assume that the message space is the type space $S$ with the interpretation that a type $s$ sending a message $s'$ is $s$ masquerading as $s'$. I refer to $M$ as an \textit{evidence system} of the environment and take it as a primitive. Let $M^{\pi}$ denote the (\textit{signal-contingent}) \textit{evidence structure} when the chosen signal is $\pi$. For any realization $s$ of a signal $\pi$, $M^{\pi}$ assigns a set $M^{\pi}(s)\subset S$ to $s$ as the set of feasible messages. Throughout the paper, I assume that $M^{\pi}$ satisfies the following conditions for any $\pi$:
\begin{itemize}
    \item[(C1)] for any $s\in \pi$,
    \begin{equation*}
        s\subset m\textup{ for all }m\in M^{\pi}(s).
    \end{equation*}
    \item[(C2)] for any $s\in \pi$,
    \begin{equation*}
        s\in M^{\pi}(s).
    \end{equation*}
    \item[(C3)] for any $s\in \pi$, and for any $\pi'\in \Pi$ that refines $s$ with $s=\cup_{s'\in \mathcal{P}}s'$ and $\mathcal{P}\subset \pi'$,
    \begin{equation*}
        M^{\pi}(s)\subset M^{\pi'}(s')\textup{ for all }s'\in \mathcal{P}.
    \end{equation*}
\end{itemize}

(C1) requires that if a subset of $\Omega\times [0,1]$ is reported at some realization, the true realization must be contained within that subset. This is a \textit{partial disclosure} setting where Sender can speak `nothing but the truth', but is not constrained to having to speak the `whole truth'. (C2) says that Sender can always report his exact type. (C3) states that any type of a signal can always be mimicked by any of its lower element of another signal that refines the type.

It is worth noting that Receiver does not observe what signal Sender chose and Sender cannot prove that.\footnote{In contrast, \cite{DeMarzoKremerSkrzypacz:19} and \cite{Shishkin:23} give Sender the ability to certify his signal.} Moreover, (C3) rules out \textit{perfect self-certification} by Sender, since a lower element (from another signal) can always masquerade.\footnote{The only exception is when the realized type takes the form of $\omega\times\{x\}$ where $\{x\}\subset [0,1]$ is a singleton. This, however, is a measure-zero set.} Taken together, Sender in the alternate game has a rather limited ability to prove what he has privately learned. Nevertheless, \textit{partial certification} is still feasible: by sending a message $m$ with $s\subset m$, a type $s$ can at least certify that the true state belongs to the projection $\proj_{\Omega}(m)$ of $m$ onto the state space $\Omega$.

\paragraph{Hierarchical Evidence Structures.} Given a signal $\pi$, a special case of evidence structures corresponds to a so-called \textit{hierarchical evidence structure} $M^{\pi}$, which, in addition to (C1)-(C3), satisfies
\begin{itemize}
    \item[(C4)] for any $s\in \pi$, $M^{\pi}(s)\subset 2^{\pi}$.
    \item[(C5)] for any $s\in \pi$, either $m\subset m'$ or $m'\subset m$ for all $m,m'\in M^{\pi}(s)$.
    \item[(C6)] for any $s\in \pi$, and for any $m\in M^{\pi}(s)$ such that $s'\in \pi$ and $s'\subset m$, it must be that $m\in M^{\pi}(s')$.%each branch can mimic the mother
\end{itemize}

(C4) requires that the set of available messages to each type $s\in \pi$ is a subset of the power set of $\pi$. (C5) says that any two available messages to each type must be nested. Finally, (C6) says that if a message is available to one type and contains another type, then it must also be available to the latter type. 

From now on, I will use the term `evidence structure' to describe general evidence structures that satisfy (C1)-(C3), and the term `hierarchical evidence structure' to describe cases that additionally satisfy (C4)-(C6).

\paragraph{Strategies and Equilibrium.} Let $\sigma = (\pi,(\gamma^{\pi'})_{\pi'\in \Pi})$ denote a strategy for Sender. This consists of a choice of signal $\pi\in \Pi$, and a messaging policy $\gamma^{\pi'}:S\rightarrow \Delta(S)$ following each signal $\pi'$ with the property that every supported message lies in $M^{\pi'}(s)$ for each $s\in \pi'$. Let $\Tilde{\mu}(m)\in \Delta(\Omega)$ denote Receiver's belief about $\omega$ when she sees a message $m$. 

I use the notion of weak sequential equilibria of \cite{Myerson:91}. They are defined as equilibria that satisfy sequential rationality and weak belief consistency, where weak consistency means Bayesian consistency on the equilibrium path and off-path beliefs that are consistent with evidence. Formally, a pair $(\sigma,\Tilde{\mu})$ is an equilibrium if $\Tilde{\mu}$ satisfies weak belief consistency, and $\sigma=(\pi,(\gamma^{\pi'})_{\pi'\in\Pi})$ is a best response to $\Tilde{\mu}$ at every information set; that is,
\begin{itemize}
    \item[(a)] Weak belief consistency: the belief map $\Tilde{\mu}:S\rightarrow \Delta(\Omega)$ is consistent with evidence and is formed using Bayes rule whenever possible. In particular, if $m$ is an on-path message, then
    \begin{align*}
        \Tilde{\mu}(m)(\omega) = \frac{\mu_0(\omega)\Big[\sum_{\{s\in \pi:m\in\supp(\gamma^{\pi}(s))\}}Pr(s|\omega)\gamma^{\pi}(s)(m)\Big]}{\sum_{\omega'\in\Omega}\mu_0(\omega')\Big[\sum_{\{s\in \pi:m\in\supp(\gamma^{\pi}(s))\}}Pr(s|\omega')\gamma^{\pi}(s)(m)\Big]}\quad \textup{for all }\omega.
    \end{align*}
    If $m$ is an off-path message, then
    \begin{equation*}
        \supp(\Tilde{\mu}(m))\subset \proj_{\Omega}(m).\footnote{
        Note that for any belief $\mu'\in \Delta(\Omega)$ that satisfies $\supp(\mu')\subset \proj_{\Omega}(m)$, it is always possible to find an element $s'\subset \Omega\times [0,1]$ such that $s'\subset m$ and $\mu' = \mu_{s'}$.
        }
    \end{equation*}
    \item[(b)] Sender ex post optimality: for any $\pi'\in \Pi$, every $s'\in \pi'$ has $\gamma^{\pi'}(s')$ supported on 
    \begin{align*}
        \argmax_{m\in M^{\pi'}(s')} \hat{v}(\Tilde{\mu}(m)).
    \end{align*}
    \item[(c)] Sender ex ante optimality: $\pi$ is supported on
    \begin{align*}
        \argmax_{\pi'\in \Pi}\sum_{\omega,s,m}\mu_0(\omega)Pr(s|\omega)\gamma^{\pi'}(s)(m)\hat{v}(\Tilde{\mu}(m)).
    \end{align*}
\end{itemize}

Throughout the paper, I restrict Sender to \textit{pure} strategies in both the choice of a signal and the messaging policy, on and off the equilibrium path. I define the \textit{outcome} of the game to be the joint distribution of the state of the world, Receiver's belief, Receiver's action, and both players' payoffs. An outcome is an \textit{equilibrium outcome} if it corresponds to an equilibrium.

\subsection{Static Reproducibility of Sequential Experimentation}

In this subsection, I first show that a hierarchical evidence structure can be naturally induced by a sequence of signals that is increasing in the refinement order. Then, I consider a so-called hierarchical alternate game where Sender chooses a refinement-ordered signal sequence. I establish that the original game and the hierarchical alternate game are disclosure equivalent.

\paragraph{Refinement-Ordered Signal Sequences.} A refinement-ordered signal sequence $\psi$ is given as a sequence of partitions $(\pi_1,\pi_2,\dots,\pi_t)$ of some length $t\geq 0$, where each partition $\pi_i$ is finer than the previous partition $\pi_{i-1}$, i.e., $\pi_i\trianglerighteq \pi_{i-1}$. Let $\Psi$ be the set of all refinement-ordered signal sequences.

Note that every refinement-ordered signal sequence $\psi = (\pi_1,\dots,\pi_t)$ induces a unique hierarchical evidence structure, denoted by $M^{\psi}$: 
\begin{align*}
    M^{\psi}(s) \equiv \Big\{\pi_1(s),\pi_2(s),\dots,\pi_t(s)\Big\}\quad \textup{for all }s\in \pi_t,
\end{align*}
where $\pi_i(s)$ is the partition element of $\pi_i$ which includes $s$. It is easy to see that every such $M^{\psi}$ satisfies (C1)-(C6). Note that a refinement-ordered signal sequence $\psi = (\pi_1,\dots,\pi_t)$ is fully determined by the last signal $\pi_t$ and its induced hierarchical evidence structure $M^{\psi}$. I call the game in which the sender chooses a refinement-ordered signal sequence, or equivalently, a single signal together with a hierarchical evidence structure, the `\textit{hierarchical alternate game}'.

The proposition below shows that the original game and the hierarchical alternate game are \textit{disclosure equivalent}; that is, for every history-dependent experimentation strategy in the original game there exists a refinement-ordered signal sequence in the hierarchical alternate game, both of which lead to the same Sender information structure and message availabilities, and vice versa.

%consequentialism...

\begin{proposition}\label{prop:disclosure_equivalence}
    The original game and the hierarchical alternate game are disclosure equivalent.
\end{proposition}

To avoid some tedious complications that do not add much insight, I omit the proof of this proposition. Instead, I use the previous professor-dean example, along with some simple diagrams, to illustrate the main idea. Recall that in the example, the professor starts off by running the experiment $\lambda_1$. If the experiment outcome is $b_1$, he stops experimenting. If the experiment outcome is $g_1$, he runs an additional experiment $\lambda_2$ with two possible outcomes $b_2$ and $g_2$. Figure \ref{fig:sequential} depicts the whole experimenting process. The truncation-feasible messages for each history are given by
\begin{align*}
    N(\{(\lambda_1,b_1)\}) &= \Big\{\emptyset, (\lambda_1,b_1)\Big\},\\
    N(\{(\lambda_1,g_1), (\lambda_2,b_2)\}) &= \Big\{\emptyset,\{(\lambda_1,g_1)\}, \{(\lambda_1,g_1), (\lambda_2,b_2)\}\Big\},\\
    N(\{(\lambda_1,g_1), (\lambda_2,g_2)\}) &= \Big\{\emptyset,\{(\lambda_1,g_1)\}, \{(\lambda_1,g_1), (\lambda_2,g_2)\}\Big\}.
\end{align*}

\begin{figure}[H]
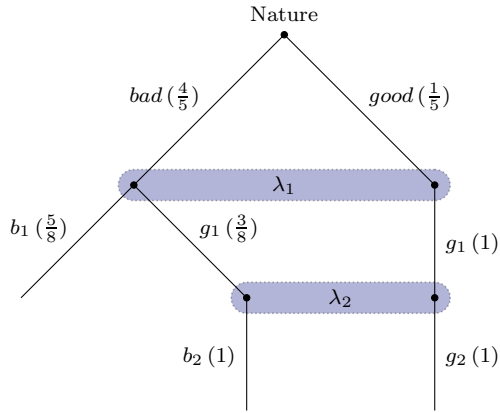

\centering
\begin{istgame}[font=\scriptsize]
\xtdistance{20mm}{40mm}
\istroot(zero){Nature}
\istb{bad\,(\frac{4}{5})}[al]
\istb{good\,(\frac{1}{5})}[ar]
\endist
\xtdistance{15mm}{30mm}
\istroot(one)(zero-1)%<left>{$\lambda_1$}
\istb{b_1\,(\frac{5}{8})}[al] 
\istb{g_1\,(\frac{3}{8})}[ar]
\endist
\istroot(two)(zero-2)
\istb{g_1\,(1)}[r] 
\endist
\xtInfosetO[fill=myblue,opacity=.3](one)(two){$\lambda_1$}
\istroot(four)(one-2)%<left>{$\lambda_2$}
\istb{b_2\,(1)}[l]
\endist
\istroot(five)(two-1)%<right>{$\lambda_2$}
\istb{g_2\,(1)}[r]
\endist
\xtInfosetO[fill=myblue,opacity=.3](four)(five){$\lambda_2$}
\end{istgame}
\caption{The professor-dean example.}
\label{fig:sequential}
\end{figure}

To reproduce Sender's information structure and truncated message space, I define a refinement-ordered signal sequence $\psi = (\pi_1,\pi_2)$ as depicted in the right panel of Figure \ref{fig:hierarchy}.\footnote{More specifically, $s_1 = bad\times [0,5/8]$, $s_2 = bad\times [5/8,1]$, and $s_3 = good\times [0,1]$.} An alternative illustration of $\psi$ is plotted as a partition tree in the left panel of Figure \ref{fig:hierarchy}. The induced hierarchical evidence structure $M^{\psi}$ can be represented by:
\begin{align*}
    M^{\psi}(s_1) &= \Big\{s_1,s_1\cup s_2\cup s_3\Big\},\\
    M^{\psi}(s_2) &= \Big\{s_2, s_2\cup s_3, s_1\cup s_2\cup s_3\Big\},\\
    M^{\psi}(s_3) &= \Big\{s_3, s_2\cup s_3, s_1\cup s_2\cup s_3\Big\}.
\end{align*}

%game tree v.s. probability tree

\begin{figure}[H]
\centering
\begin{minipage}{0.45\linewidth}
\begin{tikzpicture}[node distance=2cm]
  \node[circle, draw, fill=black,inner sep=2pt, label=above:{$\{s_1,s_2,s_3\}$}] (123) {};
  \node[circle, draw, fill=black,inner sep=2pt, label=below:{$\{s_1\}$}] (1) [below left of=123] {};
  \node[circle, draw, fill=black,inner sep=2pt, label=right:{$\{s_2,s_3\}$}] (23) [below right of=123] {};
  \node[circle, draw, fill=black,inner sep=2pt, label=below:{$\{s_2\}$}] (2) [below left of=23] {};
  \node[circle, draw, fill=black,inner sep=2pt, label=below:{$\{s_3\}$}] (3) [below right of=23] {};
  \draw[-stealth,shorten >=2pt, shorten <=2pt,line width=1pt] (123) -- (1);
  \draw[-stealth,shorten >=2pt, shorten <=2pt,line width=1pt] (123) -- (23);
  \draw[-stealth,shorten >=2pt, shorten <=2pt,line width=1pt] (23) -- (2);
  \draw[-stealth,shorten >=2pt, shorten <=2pt,line width=1pt] (23) -- (3);
\end{tikzpicture}
\end{minipage}
\begin{minipage}{0.45\linewidth}
\begin{tikzpicture}[mydrawstyle/.style={draw=black, thick}, x=1mm, y=1mm, z=1mm,scale=.7]
    \draw (-10,70) node {\scriptsize $x$};
    \draw[mydrawstyle, -](-5,70)--(35,70) node at (15,72)[above=1mm]{\scriptsize$\omega = bad$};
    \draw[mydrawstyle](-5,68)--(-5,72) node[below=4mm]{\scriptsize$0$};
    \draw[mydrawstyle](35,68)--(35,72) node[below=4mm]{\scriptsize$1$};
    \draw[mydrawstyle, -](50,70)--(90,70) node at (70,72)[above=1mm]{\scriptsize$\omega = good$};
    \draw[mydrawstyle](50,68)--(50,72) node[below=4mm]{\scriptsize$0$};
    \draw[mydrawstyle](90,68)--(90,72) node[below=4mm]{\scriptsize$1$};
  
    \draw (-10,55) node {\scriptsize$\pi_1$};
    \draw[very thick, myblue] (-5,55) node[scale=1.4] {[};
    \draw[very thick, myblue] (19.5,55) node[scale=1.4] {]};
    \draw (7.25,55) node {\scriptsize$s_1$};
    \draw[very thick, myblue] (20.5,55) node[scale=1.4] {[};
    \draw[very thick, myblue] (35,55) node[scale=1.4] {]};
    \draw (27.75,55) node {\scriptsize$s_2\cup s_3$};
    \draw[draw=none] (20.5,50.9) rectangle (35,59.1);
    \draw[very thick, myblue] (50,55) node[scale=1.4] {[};
    \draw[very thick, myblue] (90,55) node[scale=1.4] {]};
    \draw (70,55) node {\scriptsize$s_2\cup s_3$};
    \draw[draw=none] (50,50.9) rectangle (90,59.1);
    
    \draw (-10,42) node {\scriptsize$\pi_2$};
    \draw[very thick, myorange] (-5,42) node[scale=1.4] {[};
    \draw[very thick, myorange] (19.5,42) node[scale=1.4] {]};
    \draw (7.25,42) node {\scriptsize$s_1$};
    \draw[very thick, myorange] (20.5,42) node[scale=1.4] {[};
    \draw[very thick, myorange] (35,42) node[scale=1.4] {]};
    \draw (27.75,42) node {\scriptsize$s_2$};
    \draw[draw=none] (20.5,37.9) rectangle (35,46.1);
    \draw[very thick, myorange] (50,42) node[scale=1.4] {[};
    \draw[very thick, myorange] (90,42) node[scale=1.4] {]};
    \draw (70,42) node {\scriptsize$s_3$};
    \draw[draw=none] (50,37.9) rectangle (90,46.1);
\end{tikzpicture}
\end{minipage}
\caption{A refinement-ordered signal sequence.}
\label{fig:hierarchy}
\end{figure}
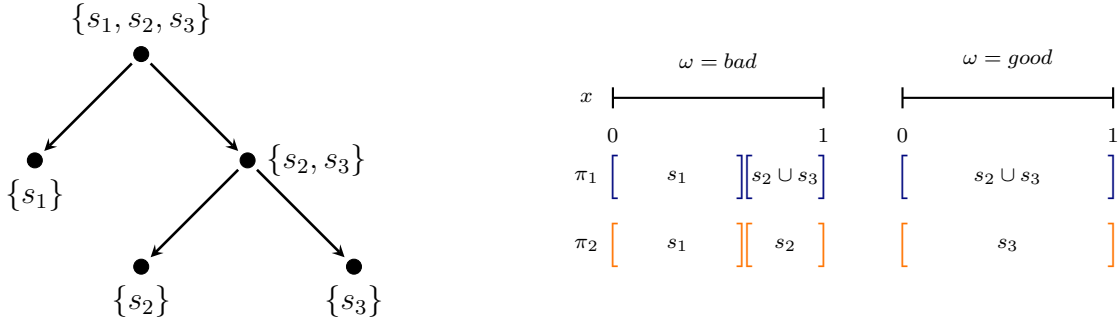

It is easy to see that the experimentation plan with $\lambda_1$ and $\lambda_2$ in the original game and the refinement-ordered signal sequence $\psi = (\pi_1,\pi_2)$ (or equivalently, the signal $\pi_2$ together with the hierarchical evidence structure $M^{\psi}$) in the hierarchical alternate game yield the same information structure for Sender. Moreover, there is a one-to-one mapping between $N(\cdot)$ and $M^{\psi}(\cdot)$, meaning that Sender has the same ability to prove his private information.
%need more words to argue having the same ability of proving...

Proposition \ref{prop:disclosure_equivalence} allows me to abstract from a more complex analysis of the original sequential experimentation game, and instead to focus on the static hierarchical alternate game. In the hierarchical alternate game, Sender essentially chooses a single signal together with a hierarchical evidence structure, which is a special case of the alternate game considered in the previous subsection. 

Henceforward, I focus on the alternate game with general evidence structures satisfying (C1)-(C3). I will show that the additional conditions (C4)-(C6) for hierarchical evidence structures do not affect the set of equilibrium outcomes.%need to rephrase "special case and shit... it is more of an example choice in the alternate game..."

\section{Equilibrium Analysis of the Alternate Game}

In this section, I impose several restrictions on equilibrium strategies and beliefs and show that it is indeed without loss of equilibrium outcomes.

An equilibrium $(\sigma,\Tilde{\mu})$ is said to be \textit{fully revealing} if Sender perfectly reveals his type on the path of play.

\begin{lemma}\label{lma:fre}
    Restricting to fully revealing equilibria does not change the set of equilibrium outcomes.
\end{lemma}

The formal proof is omitted for brevity. The basic idea behind this lemma, however, is simple. If there are multiple types sending the same message on the equilibrium path, then one can simply coarsen the signal by pooling together these types into a single type. Under the new signal, there exists a fully revealing equilibrium, and moreover, the induced outcome is exactly the same as that in the original equilibrium.%this is why I omit sender's belief as an element of outcomes...

In a fully revealing equilibrium $(\pi,(\gamma^{\pi'})_{\pi'\in \Pi},\Tilde{\mu})$, Sender and Receiver will end up sharing a common posterior belief. Noting that $\gamma^{\pi}(s)=s$ is always feasible by (C2), I can restrict attention to the case where $\gamma^{\pi}(s)=s$ for all $s\in \pi$.

I also impose one substantive restriction on beliefs off the equilibrium path. In particular, I assume that beliefs arising from off-path messages are those that result in the minimum utility for Sender, subject to the constraint that they are consistent with the message that was sent.

\begin{definition}
    A belief map $\Tilde{\mu}$ is said to be worst-off-path-punishment (WOPP) if for any given off-path message $m$, 
    \begin{align}\label{eq:wopp}
        \Tilde{\mu}(m)\in \argmin_{\mu\in\big\{\mu'\in \Delta(\Omega):\supp(\mu')\subset \proj_{\Omega}(m)\big\}}\hat{v}(\mu).
    \end{align}
\end{definition}

Notice, first, that for any off-path message, a worst punishment from a WOPP belief map does not depend on Sender's actual type. Under the assumption that Sender's preference is state-independent, this implies that Receiver does not need to learn about Sender's type to impose the harshest credible punishment. Second, for any off-path message $m$, the set on the right-hand side of (\ref{eq:wopp}) does not depend on $M$. That being said, any evidence system $M$ and its signal-contingent evidence structure satisfying (C1)-(C3) will lead to the same set of worst punishment beliefs at $m$.

The next result shows that restricting to the strongest credible punishment off the equilibrium path is without loss.

\begin{lemma}\label{lma:off-path-worst}
    Requiring beliefs to be WOPP does not change the set of equilibrium outcomes.
\end{lemma}

\begin{proof}
    It will suffice to show that given any equilibrium $(\pi,(\gamma^{\pi'})_{\pi'\in\Pi},\Tilde{\mu})$, replacing the belief map $\Tilde{\mu}$ by the following $\Tilde{\mu}'$ leaves the equilibrium outcome unchanged: $\Tilde{\mu}'(m) = \Tilde{\mu}(m)$ if $m$ is on-path, and $\Tilde{\mu}'(m)\in \argmin_{\mu\in\big\{\mu'\in \Delta(\Omega):\supp(\mu')\subset \proj_{\Omega}(m)\big\}}\hat{v}(\mu)$ if it is off-path. To see this, note that for any (on- and off-path) type who chose not to send an off-path message when facing $\Tilde{\mu}$, it is still incentive compatible not to do so when facing the WOPP belief rule $\Tilde{\mu}'$. Thus, $(\pi,(\gamma^{\pi'})_{\pi'\in\Pi},\Tilde{\mu}')$ still constitutes an equilibrium and also generates the same equilibrium outcome.
\end{proof}%note that here, sender state independent preference is important...

Relying on Lemma \ref{lma:off-path-worst}, I will focus on equilibria with WOPP belief maps. Given this punishment rule, Lemma \ref{lma:truth_or_onpath} below shows that, when Sender deviates to another signal and obtains an off-path signal realization, it is without loss of generality to consider only a simple class of messaging policies.

\begin{lemma}\label{lma:truth_or_onpath}
    Let $(\pi,(\gamma^{\pi'})_{\pi'\in \Pi},\Tilde{\mu})$ denote an equilibrium. Then for any signal deviation $\pi'$ and its element $s'\in \pi'$, without loss of generality, I can restrict attention to the case where $\gamma^{\pi'}(s')$ is supported on $\{s'\}\cup (\pi\cap M^{\pi'}(s'))$. That is, the type $s'$ either fully reveals his type, or masquerades as an on-path type, if feasible.
\end{lemma}

\begin{proof}
    If $\gamma^{\pi'}(s')\in M^{\pi'}(s')\setminus \pi$, then the type $s'$ must send an off-path message. Following that $\Tilde{\mu}$ is WOPP and that $\proj_{\Omega}(s')\subset \proj_{\Omega}(m)$ due to (C1), it is easy to see that Sender weakly prefers reporting $s'$ to reporting any $m\in M^{\pi'}(s'))\setminus \pi$:
    \begin{align*}
        \min_{\mu\in\{\mu'\in \Delta(\Omega): \supp(\mu')\subset \proj_{\Omega}(s')\}}\hat{v}(\mu)\geq \min_{\mu\in\{\mu'\in \Delta(\Omega): \supp(\mu')\subset \proj_{\Omega}(m)\}}\hat{v}(\mu).
    \end{align*}
\end{proof}

Thus, Lemma \ref{lma:truth_or_onpath} gives a simple characterization of Sender's incentives given off-path private information (or type): either fully revealing his type, or masquerading as a certain on-path type.

Since Sender's choice of a signal is covert, any deviation acquiring a different signal is undetectable. Moreover, because the set of all signals is large, pinning down signals that Sender might potentially deviate to can be difficult. The next result, however, simplifies the analysis considerably: one only needs to check whether Sender wants to deviate to signals that are refinement-ordered.

\begin{lemma}\label{lma:refinement_deviation}
    Let $(\pi,(\gamma^{\pi'})_{\pi'\in \Pi})$ and $\Tilde{\mu}$ denote a strategy profile and a belief rule. Then, if there are no profitable deviations from the signal $\pi$ to another signal which is a refinement of $\pi$, then there are no profitable deviations from $\pi$ to any signal.
\end{lemma}

\begin{proof}
    Let $\pi'$ denote a different signal which is not a refinement of $\pi$. It will suffice to show that Sender can do weakly better by switching from $\pi'$ to another signal $\pi''$ that is a refinement of $\pi$. By definition of $\pi'$, there exists a type $s'\in \pi'$ satisfying $s'\notin \pi$. Then there must exist a subset of $\pi$, denoted by $\mathcal{P}\subset \pi$, such that $s\cap s'\neq \emptyset$ for all $s\in \mathcal{P}$, and $s'\subset \cup_{s\in \mathcal{P}}s$. This implies that any message $m\in M^{\pi'}(s')$ is off-path due to (C1). By Lemma \ref{lma:truth_or_onpath}, the type $s'$ weakly prefers to fully reveal himself. Notice that $\hat{v}(\Tilde{\mu}(s'))\leq \hat{v}(\Tilde{\mu}(s))$ for all $s\in \mathcal{P}$ with $\Tilde{\mu}$ being WOPP. Hence, Sender can do weakly better by replacing the signal realization $s'$ with the corresponding set of signal realizations $\{s\cap s'\}_{s\in \mathcal{P}}$. Repeat this process until every such type $s'$ is exhausted in $\pi'$, which generates a new signal, denoted by $\pi''$. It is easy to see that Sender weakly prefers $\pi''$ to $\pi'$, and moreover, $\pi''$ is a refinement of $\pi$, completing the proof.\footnote{More precisely, $\pi''$ is, by construction, a refinement of $\pi'$ as well. Indeed, I can write $\pi''$ as $\pi'' = \pi\lor \pi'$, where $\lor$ denotes the join, that is, $\pi\lor\pi'$ is the coarsest refinement of both $\pi$ and $\pi'$.}
\end{proof}

Lemma \ref{lma:refinement_deviation} identifies the only class of deviations that needs to be checked at Sender's signal choice stage. This allows me to greatly simplify the analysis of the game. I simplify the analysis further by deriving a sufficient and necessary condition for the existence of a profitable signal deviation.

\begin{lemma}\label{lma:profitable_deviation_iff}
   Let $(\pi,(\gamma^{\pi'})_{\pi'\in \Pi})$ and $\Tilde{\mu}$ denote a strategy profile and a belief rule. Then, $\pi'$ is a profitable signal deviation from $\pi$ with $\pi'\trianglerighteq \pi$ if and only if $\pi'$ entails a signal realization $s'\subsetneq s\in \pi$ such that $\hat{v}(\Tilde{\mu}(s'))>\hat{v}(\Tilde{\mu}(s))$. Moreover, $\proj_{\Omega}(s')$ must be a strict subset of $\proj_{\Omega}(s)$.
\end{lemma}

The formal proof is omitted for brevity. The basic idea of deviating from $\pi$ to $\pi'$ is simple. To illustrate, suppose that $\pi'$ entails two signal realizations $s'$ and $s''$ such that $s = s'\cup s''\in \pi$ and $\hat{v}(\Tilde{\mu}(s'))>\hat{v}(\Tilde{\mu}(s))$. Given $\pi'$, when the signal realization is $s'$, Sender reports $m=s'$, which makes the deviation detected and induces a better action even Receiver will impose the harshest punishment. When the signal realization is $s''$, Sender reports $m=s$, which is feasible due to (C3), leaving the deviation undetected so Receiver still acts as if Sender's true type is $s$.

Combining Lemma \ref{lma:fre}-\ref{lma:profitable_deviation_iff} allows me to, \textit{without loss of equilibrium outcomes}, restrict attention to a particular class of equilibria of the form $(\pi,(\gamma^{\pi'})_{\pi'\in \Pi},\Tilde{\mu})$ with the following properties:
\begin{itemize}
    \item[(P1)] For all $m\in \pi$,
    \begin{align*}
        \Tilde{\mu}(m)(\omega) = \frac{\mu_0(\omega)\ell(\{x|(\omega,x)\in s\})}{\sum_{\omega'\in\Omega}\mu_0(\omega')\ell(\{x|(\omega',x)\in s\})}\quad \textup{for all } \omega.
    \end{align*}
    For all $m\notin \pi$, 
    \begin{align*}
        \Tilde{\mu}(m)\in \argmin_{\mu\in\{\mu'\in \Delta(\Omega):\supp(\mu')\subset \proj_{\Omega}(m)\}}\hat{v}(\mu).
    \end{align*}
    \item[\textup{(P2)}] For all $s\in \pi$, $\gamma^{\pi}(s) = s$ and it is supported on
    \begin{align*}
        \argmax_{m\in M^{\pi}(s)}\hat{v}(\Tilde{\mu}(m)).
    \end{align*}
    \item[(P3)] There does not exist $s'\subsetneq s$ with $\proj_{\Omega}(s')\subsetneq \proj_{\Omega}(s)$ such that $\hat{v}(\Tilde{\mu}(s'))>\hat{v}(\Tilde{\mu}(s))$.
\end{itemize}

From this point onwards, I focus on equilbira satisfying (P1)-(P3).

\section{The Belief-Based Approach}\label{sec:belief-based}

Recall that Lemma \ref{lma:profitable_deviation_iff} provides a simple necessary and sufficient condition for the existence of a profitable deviation at the signal choice stage. Specifically, for a given strategy profile $(\pi,(\gamma^{\pi'})_{\pi'\in \Pi})$ and belief rule $\Tilde{\mu}$, Sender would not covertly deviate to a different signal if there does not exist a realization $s\in\pi$ and a lower element $s'\subsetneq s$ that satisfies both $\proj_{\Omega}(s')\subsetneq \proj_{\Omega}(s)$ and $\hat{v}(\Tilde{\mu}(s'))>\hat{v}(\Tilde{\mu}(s))$.

It turns out that this condition can be equivalently restated in the belief space. To see this, let $\mu$ denote Receiver's posterior belief when some type $s$ from a signal $\pi$ chooses to fully reveal his type. Correspondingly, let $\mu'$ denote Receiver's posterior belief if Sender deviates to another signal with some realization $s'\subset s$ and chooses to truthfully reveal it. Then the deviation is not profitable if $\hat{v}(\mu)\geq \hat{v}(\mu')$ for any such $\mu'$ with $\supp(\mu')\subsetneq \supp(\mu)$. 

Therefore, the absence of profitable deviations from a signal is equivalent to requiring its induced posteriors have the above-mentioned property which I term `additional-learning-proof' below. 

\begin{definition}
    A belief $\mu\in \Delta(\Omega)$ is called additional-learning-proof (ALP) if
    \begin{align*}
        \hat{v}(\mu) \geq \argmin_{\mu'\in \big\{\mu''\in \Delta(\Omega):\supp(\mu'')\subsetneq \supp(\mu)\big\}} \hat{v}(\mu').
    \end{align*}
\end{definition}

In other words, a belief $\mu$ satisfies the ALP property if Sender cannot strictly benefit from secretly learning a higher refinement-ordered signal, and subsequently `surprising' Receiver only when good news arrives (sending off-path messages that strictly improves upon his on-path payoff). Let $\Gamma^{ALP}\subset\Delta(\Omega)$ denote the set of all additional-learning-proof beliefs.

\begin{comment}
\begin{conjecture}
    The set of ALP beliefs $\Gamma^{ALP}$ is equal to
    \begin{align*}
        \Gamma^{ALP} = \Big\{\mu\in \Delta(\Omega):\hat{v}(\mu)\geq \hat{v}(\delta_{\omega}),\textup{ } \forall \omega\in \supp(\mu)\Big\},
    \end{align*}
    where $\delta_{\omega}$ is the degenerate belief about the state $\omega$.
\end{conjecture}
\end{comment}

The above discussion suggests analyzing the model via the `belief-based approach', as is common in the communication literature.\footnote{For example, see \cite{KamenicaGentzkow:11} and \cite{LipnowskiRavid:20}.} This approach uses the ex ante distribution over Receiver's posterior beliefs as a \textit{substitute} for both Sender's strategy and the belief system. Clearly, every strategy profile $(\pi,(\gamma^{\pi'})_{\pi'\in \Pi})$ and belief rule $\Tilde{\mu}$ generate such a distribution over Receiver's posterior beliefs, $\tau\in \Delta(\Delta(\Omega))$. By Bayes's rule, this distribution of posteriors averages to the prior, $\mu_0$; that is, $\sum_{\supp(\tau)}\mu\tau(\mu)=\mu_0$. Moreover, Lemma \ref{lma:profitable_deviation_iff} implies that the support of $\tau$ must lie within the set $\Gamma^{ALP}$ if $\tau$ is induced by some equilibrium. The following main theorem of the article formally states this insight, and shows that the converse direction is true as well.

 Let $(\tau,\nu)$ denote a pair representing Receiver's posterior distribution, $\tau$, and Sender's ex ante expected payoff, $\nu$.

\begin{theorem}\label{thm:belief-based-characterize}
    If $(\pi,(\gamma^{\pi'})_{\pi'\in \Pi},\Tilde{\mu})$ is an equilibrium, then the induced pair $(\tau,\nu)$ must satisfy 
    \begin{itemize}
        \item[\textup{(}I\textup{)}] $\tau$ is a mean-preserving spread of $\mu_0$ supported on $\Gamma^{ALP}$, and
        \item[\textup{(}II\textup{)}] $\nu = E_{\tau}\hat{v}(\mu)$.
    \end{itemize}
    Conversely, if a pair $(\tau,\nu)$ satisfies \textup{(}I\textup{)} and \textup{(}II\textup{)}, then there must exist an equilibrium that induces it.
\end{theorem}
%for the converse direction, for any pair, need to construct an equilibrium...

The key observation behind Theorem \ref{thm:belief-based-characterize} is that one can transform any Bayes plausible distribution of posteriors supported on $\Gamma^{ALP}$ into an equilibrium signal. The theorem also yields a convenient formula to determine Sender's maximal equilibrium value---one only needs to search over distributions of posteriors satisfying ($I$), as presented below.
\begin{align*}
    \max_{\tau\in \Gamma^{ALP}}E_{\tau}\hat{v}(\mu)\\
    \textup{s.t.}\quad \sum_{\supp(\tau)}\mu\tau(\mu)=\mu_0.
\end{align*}

Extending the concavification method of \cite{AumannMaschler:95} and \cite{KamenicaGentzkow:11}, the above optimaization problem leads to a useful geometric characterization of Sender's maximal equilibrium value, which I present in Corollary \ref{cor:concavification} below. A few more definitions are required. Let $W\subset \Delta(\Omega)$ be any subset of the belief space. Define $V_W$ to be the $W$-\textit{concavification} of $\hat{v}$, namely, the smallest concave function that is pointwise larger than $\hat{v}$ on the domain $W$. Let $\nu^*$ denote Sender's maximal equilibrium value.

\begin{corollary}\label{cor:concavification}
    Sender's maximal equilibrium value is given by $\hat{v}$'s concave envelope supported upon $\Gamma^{ALP};$ that is,
    \begin{align*}
        \nu^* = V_{\Gamma^{ALP}}(\mu_0).
    \end{align*}
\end{corollary}

\section{The Binary-State Case}

In Section \ref{sec:belief-based}, I have identified the set of beliefs that can be induced in an equilibrium, namely, $\Gamma^{ALP}$. However, since this set is sensitive to details of the set of preferences, it can be difficult to characterize in cases when the state space is large. In this section, I consider a  binary-state version of the model. This allows me to visualize Sender's indirect utility, $\hat{v}$, the set of ALP beliefs, $\Gamma^{ALP}$, and the $\Gamma^{ALP}$-concavification of $\hat{v}$, $V_{\Gamma^{ALP}}$. Then, I can illustrate how Sender achieves his maximal equilibrium value in a more intuitive manner.

Let $\Omega = \{L,R\}$. Since the state space is binary, I abuse notation by associating each belief $\mu$ with $Pr(\omega=R)\in [0,1]$. The set of ALP beliefs now takes a simple form:
\begin{align*}
    \Gamma^{ALP} = \big\{0,1\big\}\cup \big\{\mu\in(0,1): \hat{v}(\mu)\geq \max\{\hat{v}(0),\hat{v}(1)\}\big\}.
\end{align*}

Figure \ref{fig:binary_example} shows an example of the construction of $V_{\Gamma^{ALP}}$. In the figure, $\mu$ denotes the probability that the state is $R$. Putting this probability on the horizontal axis, the figure plots an arbitrary indirect utility function $\hat{v}$ (left), along with its ALP belief domain (center) and the concave envelope on $\Gamma^{ALP}$.

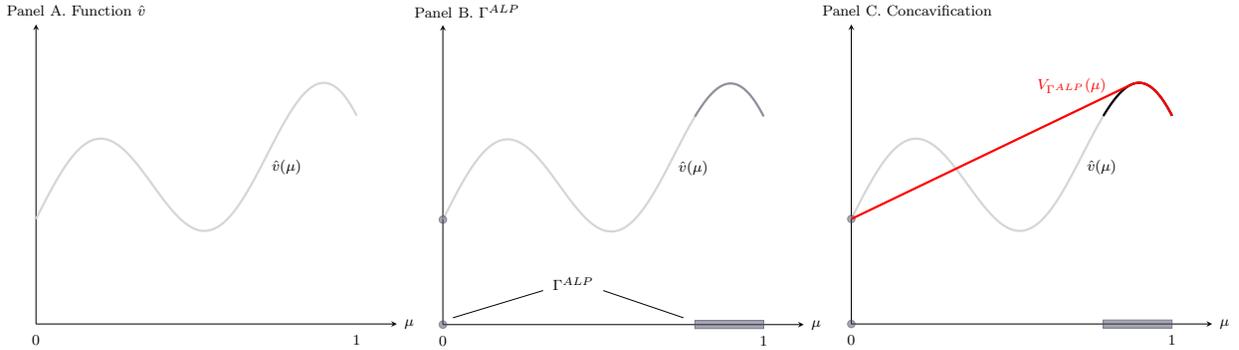
\begin{figure}[H]
\centering
\captionsetup{font=footnotesize}
\begin{minipage}{0.32\linewidth}
\begin{tikzpicture}[scale=.7]
    \begin{axis}
    [
    axis lines = left,
    xlabel={$\mu$},
    ylabel={},
    xmin=0.2, xmax=2,
    ymin=0, ymax=2.2,
    xtick=\empty,
    xticklabels=\empty,
    ytick=\empty,
    yticklabels=\empty,
    every axis x label/.style={at={(ticklabel* cs:1)},anchor=west,},
    scaled y ticks = false,every axis y label/.style={at={(ticklabel* cs:1)},anchor=south,},
    label style={font=\scriptsize},
    tick label style={font=\scriptsize},
    clip=false
    ]
    \draw (axis cs:.2,-.12) node {\scriptsize $0$};
    \draw (axis cs:1.8,-.12) node {\scriptsize $1$};
    \addplot[domain=0.2:1.8, samples=100,very thick,mylightgrey]{-1/5*cos(deg(3*x))-1/3*sin(3.5*deg(x))-1/2*cos(5*deg(x))+1/2*sin(5*deg(x))+1};
    \draw (axis cs:1.45,1.15) node {\scriptsize $\hat{v}(\mu)$};
    \draw (axis cs:.4,2.3) node {\scriptsize Panel A. Function $\hat{v}$};
    \end{axis}
\end{tikzpicture}
\end{minipage}
\begin{minipage}{0.32\linewidth}
\begin{tikzpicture}[scale=.7]
    \begin{axis}
    [
    axis lines = left,
    xlabel={$\mu$},
    ylabel={},
    xmin=0.2, xmax=2,
    ymin=0, ymax=2.2,
    xtick=\empty,
    xticklabels=\empty,
    ytick=\empty,
    yticklabels=\empty,
    every axis x label/.style={at={(ticklabel* cs:1)},anchor=west,},
    scaled y ticks = false,every axis y label/.style={at={(ticklabel* cs:1)},anchor=south,},
    label style={font=\scriptsize},
    tick label style={font=\scriptsize},
    clip=false
    ]
    \draw (axis cs:.2,-.12) node {\scriptsize $0$};
    \draw (axis cs:1.8,-.12) node {\scriptsize $1$};
    \addplot[domain=0.2:1.8, samples=100,very thick,mylightgrey]{-1/5*cos(deg(3*x))-1/3*sin(3.5*deg(x))-1/2*cos(5*deg(x))+1/2*sin(5*deg(x))+1};
    \draw (axis cs:1.45,1.15) node {\scriptsize $\hat{v}(\mu)$};

    \addplot[domain=1.457:1.8, samples=100,very thick,mygrey,opacity=.5]{-1/5*cos(deg(3*x))-1/3*sin(3.5*deg(x))-1/2*cos(5*deg(x))+1/2*sin(5*deg(x))+1};
    \draw[mygrey,fill=mygrey,opacity=.5] (axis cs:.2,.771) circle (.4ex);

    \draw[fill=mygrey,mygrey,opacity=.5] (axis cs:1.457,-.03) rectangle (axis cs:1.8,.03);
    \draw[mygrey,fill=mygrey,opacity=.5] (axis cs:.2,0) circle (.4ex);
    \draw (axis cs:.85,.3) node {\scriptsize $\Gamma^{ALP}$};
    \draw[black](axis cs:.25,.05)--(axis cs:.7,.25);
    \draw[black](axis cs:1.4,.05)--(axis cs:1.,.25);
    \draw (axis cs:.32,2.3) node {\scriptsize Panel B. $\Gamma^{ALP}$};
    \end{axis}
\end{tikzpicture}
\end{minipage}
\begin{minipage}{0.32\linewidth}
\begin{tikzpicture}[scale=.7]
    \begin{axis}
    [
    axis lines = left,
    xlabel={$\mu$},
    ylabel={},
    xmin=0.2, xmax=2,
    ymin=0, ymax=2.2,
    xtick=\empty,
    xticklabels=\empty,
    ytick=\empty,
    yticklabels=\empty,
    every axis x label/.style={at={(ticklabel* cs:1)},anchor=west,},
    scaled y ticks = false,every axis y label/.style={at={(ticklabel* cs:1)},anchor=south,},
    label style={font=\scriptsize},
    tick label style={font=\scriptsize},
    clip=false
    ]
    \draw (axis cs:.2,-.12) node {\scriptsize $0$};
    \draw (axis cs:1.8,-.12) node {\scriptsize $1$};
    \addplot[domain=0.2:1.8, samples=100,very thick,mylightgrey]{-1/5*cos(deg(3*x))-1/3*sin(3.5*deg(x))-1/2*cos(5*deg(x))+1/2*sin(5*deg(x))+1};
    \addplot[domain=1.457:1.8, samples=100,very thick]{-1/5*cos(deg(3*x))-1/3*sin(3.5*deg(x))-1/2*cos(5*deg(x))+1/2*sin(5*deg(x))+1};
    \draw[mygrey,fill=mygrey,opacity=.5] (axis cs:.2,.771) circle (.4ex);
    \draw[very thick,red](axis cs:.2,.771)--(axis cs:1.593,1.7558);
    \addplot[domain=1.593:1.8, samples=100,very thick,red]{-1/5*cos(deg(3*x))-1/3*sin(3.5*deg(x))-1/2*cos(5*deg(x))+1/2*sin(5*deg(x))+1};
    
    \draw[fill=mygrey,mygrey,opacity=.5] (axis cs:1.457,-.03) rectangle (axis cs:1.8,.03);
    \draw[mygrey,fill=mygrey,opacity=.5] (axis cs:.2,0) circle (.4ex);
    \draw (axis cs:1.45,1.15) node {\scriptsize $\hat{v}(\mu)$};
    \draw (axis cs:1.3,1.75) node {\scriptsize \textcolor{red}{$V_{\Gamma^{ALP}}(\mu)$}};
    \draw (axis cs:.48,2.3) node {\scriptsize Panel C. Concavification};
    \end{axis}
\end{tikzpicture}
\end{minipage}
\caption{An illustration of concavification in a binary-state example.}
\label{fig:binary_example}
\end{figure}

%\section{Extension: Sender State-Dependent Preferences}\label{sec:state-dependent}

\section{Conclusion}

This paper studies persuasion in a setting where Sender tries to persuade Receiver with evidence that stems from sequential private experimentation and that can be selectively revealed in a right-truncation manner. I first turn the dynamic problem into a static one. Then I use the belief-based approach to analyze the value that Sender derives from such an activity of evidence collection and disclosure.

\newpage
\appendix
\section{Appendix}\label{sec:appendix}

\begin{comment}
\subsection{Proof of Proposition \ref{prop:disclosure_equivalence}}
I first show that every history-dependent experimentation strategy in the original game defines a refinement-ordered signal sequence that lead to the same Sender information structures and message availabilities in all scenarios.
To show the opposite direction, start with a refinement-ordered signal sequence $(\pi_1,\dots,\pi_t)$ for any $t\geq 1$.
\end{comment}

\subsection{Proof of Theorem \ref{thm:belief-based-characterize}}

When $(\pi,(\gamma^{\pi'})_{\pi'\in \Pi},\Tilde{\mu})$ is an equilibrium and $\tau$ is the induced distribution of Receiver's (as well as Sender's, cf. Lemma \ref{lma:fre}) posterior beliefs, Bayesian consistency on the equilibrium path implies that $\sum_{\supp(\tau)}\mu\tau(\mu)=\mu_0$, where $\supp(\tau) = \{\mu_s\}_{s\in \pi}$. Moreover, since $(\pi,(\gamma^{\pi'})_{\pi'\in \Pi},\Tilde{\mu})$ is an equilibrium, by the argument in the beginning of Section \ref{sec:belief-based}, the posterior distribution $\tau$ must be supported on the set $\Gamma^{ALP}$. Since both players share the posterior, Sender's expected utility is equal to $\nu = E_{\tau}\hat{v}(\mu)$.

To prove the converse, let $(\tau,\nu)$ be a pair of posterior distribution and Sender's expected utility satisfying ($I$) and ($II$). I will construct an equilibrium which generates it. Since $\tau$ is a mean-preserving spread of $\mu_0$, there always exists a signal that induces it (\citealt{GreenStokey:78}). Let $\pi$ denote such a signal. Define the belief map $\Tilde{\mu}:S\rightarrow \Delta(\Omega)$ according to (P$_1$). For any $\pi'\in \Pi$, define
\[
    \gamma^{\pi'}(s') = 
    \begin{cases}
        s'&\quad\textup{if } \hat{v}(\Tilde{\mu}(s'))>\hat{v}(\Tilde{\mu}(s))\textup{ for some }s\in \pi\textup{ and }s'\subset s, \\
        s&\quad\textup{otherwise}.
    \end{cases}
\]
Now I verify that the constructed $(\pi,(\gamma^{\pi'})_{\pi'\in \Pi},\Tilde{\mu})$ is an equilibrium. First, it is straightforward that $\Tilde{\mu}$ satisfies weak belief consistency. Second, if Sender chooses the signal $\pi$, then for any realization $s\in \pi$, it must be optimal to truthfully reveal $s$ since by definition $\mu_s\in \Gamma^{ALP}$. If Sender chooses a different signal $\pi'\in \Pi\setminus\{\pi\}$, then $\gamma^{\pi'}(\cdot)$ defined above satisfies ex post optimality by Lemma \ref{lma:truth_or_onpath}. Last, using the fact that $\supp(\tau)\subset \Gamma^{ALP}$ again, one immediately obtains that $\pi$ is preferred to any other $\pi'$. Therefore, $(\pi,(\gamma^{\pi'})_{\pi'\in \Pi},\Tilde{\mu})$ constitutes an equilibrium which delivers the expected payoff $\nu = E_{\tau}\hat{v}(\mu)$. This concludes the proof.

\newpage
\bibliography{mybib}
\bibliographystyle{jpe} 

\end{document}